\documentclass[a4paper,UKenglish]{lipics-v2016}

\newcommand{\shortversion}[1]{}
\newcommand{\longversion}[1]{#1}

\usepackage{microtype}
\usepackage{amsmath,amsthm}
\usepackage{amssymb}
\usepackage{hyperref,url}
\usepackage{mathtools}
\usepackage{tikz}
\usetikzlibrary{matrix}
\usepgflibrary{shapes.arrows}
\pgfrealjobname{barrier}
\usepackage{enumerate}
\usepackage{graphicx}
\usepackage{verbatim}
\usepackage{authblk}
\usepackage{xspace}

\newtheorem{obs}{Observation}

\newcommand{\set}[1]{\left\{ #1 \right\}}

\newcommand{\abs}[1]{\left| #1 \right|}
\newcommand{\etal}{\textit{et~al.}}

\newcommand{\opt}{\textsc{opt}}

\newcommand{\ccfont}[1]{\textnormal{\textsf{#1}}\xspace}

\newcommand{\Ptime}{\ccfont{P}}
\newcommand{\NP}{\ccfont{NP}}

\newcommand{\W}[1]{\ccfont{W[#1]}}

\allowdisplaybreaks

\begin{document}

\title{Barrier Coverage with Non-uniform Lengths to Minimize Aggregate Movements\footnote{This work was supported by  the Australian Research Council (ARC) under the Discovery Projects funding scheme (DP150101134). Serge Gaspers is the recipient of an ARC Future Fellowship (FT140100048).}}

\author[1,2]{Serge Gaspers}
\author[3]{Joachim Gudmundsson}
\author[3]{Juli\'{a}n Mestre}
\author[1,3]{Stefan R\"{u}mmele}
\affil[1]{UNSW Sydney, Australia}
\affil[2]{Data61, CSIRO, Australia}
\affil[3]{The University of Sydney, Australia}

\Copyright{Serge Gaspers, Joachim Gudmundsson, Juli\'{a}n Mestre, and Stefan R\"{u}mmele}%

\maketitle

\begin{abstract}
Given a line segment $I=[0,L]$, the so-called \emph{barrier}, and a set of $n$ sensors with varying ranges positioned on the line containing $I$, the \emph{barrier coverage} problem is to move the sensors so that they cover $I$, while minimising the total movement. In the case when all the sensors have the same radius the problem can be solved in $O(n \log n)$ time (Andrews and Wang, Algorithmica 2017). If the sensors have different radii the problem is known to be \NP-hard to approximate within a constant factor (Czyzowicz \etal, ADHOC-NOW 2009).

We strengthen this result and prove that no polynomial time $\rho^{1-\varepsilon}$-approxi\-mation algorithm exists unless $\Ptime=\NP$, where $\rho$ is the ratio between the largest radius and the smallest radius. Even when we restrict the number of sensors that are allowed to move by a parameter $k$, the problem turns out to be \W{1}-hard. On the positive side we show that a $((2+\varepsilon)\rho+2/\varepsilon)$-approximation can be computed in $O(n^3/\varepsilon^2)$ time and we prove fixed-parameter tractability when parameterized by the total movement assuming all numbers in the input are integers.
\end{abstract}

\section{Introduction}
The original motivation for the problem of covering barriers comes from intrusion detection, where the goal is to guard the boundary (barrier) of a region in the plane. In this case the barrier can be described by a polygon and the initial position of the sensors can be anywhere in the plane. The barrier coverage problem, and many of its variants, has received much attention in the wireless sensor community, see for example~\cite{AroraEtal-exscal-05,ckl-lbcws-10,kla-bcws-05} and the recent surveys~\cite{tw-sbcds-14,wgwgc-sbcs-16}. Large scale barriers with more than a thousand sensors have been experimentally tested and evaluated~\cite{AroraEtal-exscal-05}.

In a general setting of the barrier coverage problem each sensor has a fixed sensor radius and is initially placed in the plane and the cost of moving a sensor is proportional to the Euclidean distance it is moved. In this paper we consider the special case where we have $n$ sensors on the real line. Each sensor $i=1,\ldots, n$ has a location $x_i$ and a radius $r_i$. When located at $y_i$, the $i$-th sensor covers the \emph{interval} $B(y_i, r_i) = [y_i - r_i, y_i + r_i]$. The goal is to move around the sensor intervals to cover the interval $[0,L]$, the so-called \emph{barrier}. In other words, for each sensor, we need to decide its new location $y_i$ so that
$
  [0,L] \subseteq \bigcup_{i} B(y_i, r_i).
$
The cost of the solution is the sum of sensor movements:
$
  \mathrm{cost}(y) = \sum_{i} |y_i - x_i|,
$
and the objective is to find a feasible solution of minimum cost.

\begin{figure}[t]
\includegraphics[width=\textwidth]{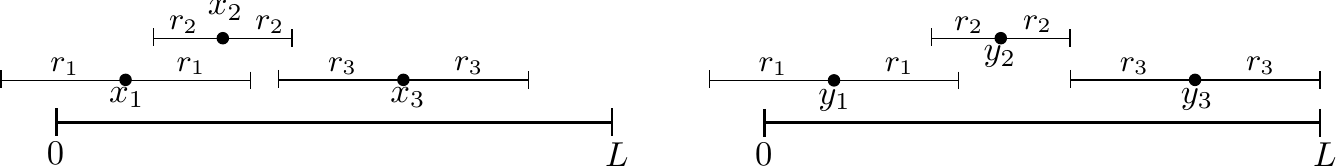}
\caption{(left) Illustrating an instance with three sensors $\{1,2,3\}$ and sensor intervals. (right) The sensors have moved such that the sensor intervals cover the barrier $[0,L]$.}
\label{fig:intro}
\end{figure}

\subsection{Our Results and Related Work}
Even though the barrier coverage problem, and many of its variants, has received a lot of attention from the wireless sensor community, not much is known from a theoretical point of view. In the literature three different optimisation criteria have been considered: minimize the sum of movements (min-sum), minimize the maximum movement (min-max) and, minimize the number of sensors that move (min-num).

Dobrev \etal~\cite{Dobrev-cbcrs-15} studied the min-sum and min-max version in the case when the sensors' start position can be anywhere in the plane and $k$ parallel barriers are required to be covered. However, they restricted the movement of the sensors to be perpendicular to the barriers. They showed an $O(kn^{k+1})$ time algorithm. If the barriers are allowed to be horizontal and vertical then the problem is \NP-complete, even for two barriers.

Most of the existing research has focussed on the special case when the barrier is a line segment $I$ and all the sensors are initially positioned on a line containing $I$.

\paragraph*{The Min-Sum model.} If all intervals have the same radius, it is not difficult to show that any solution can be converted into one where $x_i < x_j$ if and only if $y_i < y_j$ without incurring any extra cost. Czyzowicz \etal~\cite{conf/adhoc/CzyzowiczKKLNOSUY10} showed an $O(n^2)$ time algorithm for this case which was later improved to $O(n \log n)$ by Andrews and Wang~\cite{AndrewsW17}. Andrews and Wang also showed a matching $\Omega(n \log n)$ lower bound.  When the radii are non-uniform, this is not the case anymore. In fact, Czyzowicz \etal~\cite{conf/adhoc/CzyzowiczKKLNOSUY10} showed that this variant of the problem is \NP-hard, and remarked that not even a $2$-approximation is possible in polynomial time. In fact their hardness proof can be modified to show (Theorem~\ref{thm:inapproximabilty}) that no approximation factor is possible. The catch is that the instance used in the reduction needs to have some intervals that are very small and some intervals that are very large. This is a scenario that is not likely to happen in practice, so the question is whether there is an approximation algorithm whose factor depends on the ratio of the largest radius to the smallest radius.

Let $\rho$ be the ratio between the largest radius $r_{\max} = \max_i r_i$ and the smallest radius $r_{\min} = \min_i r_i$. Theorem~\ref{thm:inapproximabilty} states that no $\rho^{1-\varepsilon}$ approximation algorithm exists for any $\varepsilon>0$ unless $\Ptime=\NP$. On the positive side we show an $O(n^3/\varepsilon^2)$ time $((2+\varepsilon)\rho+2/\varepsilon)$-approximation algorithm for any given $\varepsilon>0$. The general idea is to look at ``order-preserving'' solutions, that is, solutions where the set of sensors covering the barrier maintains their individual order from left to right. This will be described in more detail in Section~\ref{sec:order-preserving}.

We also study the problem from the perspective of parameterized complexity and show that the problem is hard even if the number of intervals required to move is small, that is \W{1}-hardness with respect to parameter number of moved intervals.
Complementary, we provide a fixed-parameter tractable algorithm when the problem is parameterized by the budget, i.e., the target sum of movements.

\paragraph*{The Min-Max and Min-Num models.}
Czyzowicz \etal~\cite{conf/adhoc/CzyzowiczKKLNOSUY10} also considered min-max version of the problem, where the aim is to minimize the maximum movement. If the sensors have the same radius they gave an $O(n^2)$ time algorithm. Chen \etal~\cite{cglw-ammsm-13} improved the bound to $O(n \log n)$. In the same paper Chen \etal{} presented an $O(n^2 \log n)$ time algorithm for the case when the sensors have different radius. For the min-num version Mehrandish \etal~\cite{Mehrandish-11} showed that the problem can be solved in polynomial time using dynamic programming if the sensor radii are uniform, otherwise the problem is \NP-hard.

\section{Order-Preserving Approximations} \label{sec:order-preserving}

Let $y$ be a solution to the barrier problem. We say a subset of intervals $S \subseteq \{1, \ldots, n\}$ is \emph{active} for a solution $y$ if the intervals in $S$ alone are enough to cover the barrier. Additionally, we say that $S$ is a \emph{minimal active set} if no proper subset of $S$ is active. Notice that in an optimal solution $y$ if $y_i \neq x_i$ then $i$ must belong to a minimal active set.
Without loss of generality we assume that $x_1 \leq x_2 \leq \cdots \leq x_n$. We say a solution $y$ is \emph{order-preserving} if it has an active set $S$ such that for any $i, j \in S$ with $i < j$, we have $y_i < y_j$.

Our algorithm is based on finding a nearly optimal order-preserving solution. First we show, in Section~\ref{ssec:order-preserving}, that there always exists an order-preserving solution that is a good approximation of the optimal unrestricted solution, and prove that our analysis is almost tight. Then, in Section~\ref{ssec:algorithm}, we show how to compute a nearly optimal order-preserving solution in polynomial time.

\subsection{Quality of Order-Preserving Solutions} \label{ssec:order-preserving}

The high level idea to prove that there exists an order-preserving solution that approximates the optimal solution is to start from an arbitrary optimal solution $y$ and progressively modify the positions of two overlapping active intervals so that they are in the right order and together cover the exact same portion of the barrier, as shown in Fig.~\ref{fig:order-preserving-example}. We refer to this process as the \emph{untangling} process.

\begin{figure}[t]
      \centering
	  \begin{tikzpicture}
		\tikzstyle{interval}=[|-|,thick]
		\tikzstyle{center}=[pos=0.5,circle,fill,inner sep=1pt]
		\draw[dotted,very thick] (0,-0.5) -- (0,1.5);
	  \draw[interval] (0,0) -- node[center,label=$y_j$] {} (2,0);
	  \draw[interval] (1,1) -- node[center, label=$y_i$] {} (4,1);
	  \draw[dotted,very thick] (4,-0.5) -- (4,1.5);
	  \node[single arrow, draw, inner sep=4pt] at (6.5,0.5) {\sc \ swap};
	  \draw[dotted,very thick] (9,-0.5) -- (9,1.5);
	  \draw[interval] (11,0) -- node[center,label=$y'_j$] {} (13,0);
	  \draw[interval] (9,1) -- node[center,label=$y'_i$] {} (12,1);
	  \draw[dotted,very thick] (13,-0.5) -- (13,1.5);
	\end{tikzpicture}
	\caption{Two overlapping intervals $i$ and $j$ being swapped. After the swap the union of the intervals cover the same section of the barrier but their centers swap order.}
    \label{fig:order-preserving-example}
\end{figure}
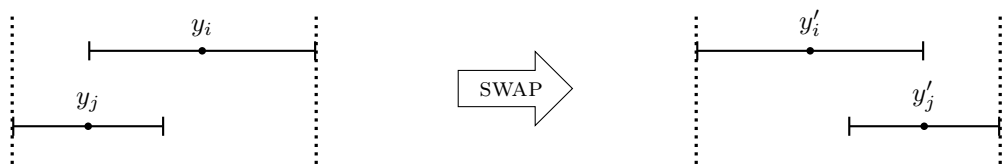

This untangling process continues until all overlapping active intervals are in order. Let us denote the resulting solution with $\hat{y}$. Our goal is to charge the cost of $\hat{y}$ to the intervals in such a way that the total charge an interval can receive is comparable to its contribution to the cost of $y$. More formally, we define an \emph{$\beta$-balanced cost sharing scheme} to be a function $\xi: S \rightarrow \mathbb{R}^+$, where
\begin{enumerate}[i)]
	\item $\mathrm{cost}(\hat{y}) \leq \sum_{i \in S} \xi(i)$, and
	\item $\xi(i) \leq \beta\, |x_i - y_i|$ for all $i \in S$.
\end{enumerate}

It is easy to see that the existence of a well balanced cost sharing scheme implies a good approximation guarantee.

\begin{lemma}\label{lem:untangling}
	Let $\hat{y}$ be the result of untangling an optimal solution $y$. If $\hat{y}$ admits an $\beta$-balanced cost sharing scheme then $\hat{y}$ is $\beta$-approximate.
\end{lemma}

\begin{proof}
	We bound the cost of $\hat{y}$ as follows:
	$\mathrm{cost}(\hat{y}) \leq \sum_{i} \xi(i) \leq \sum_i \beta |x_i - y_i| = \beta \cdot \mathrm{cost}(y) = \beta \cdot \opt$,
	where the first two inequalities follow from the definition of $\beta$-balancedness and the last equality follows from the fact that $y$ is optimal.
\end{proof}

To show the existence of a good cost sharing scheme, we will study the structure of an optimal solution $y$ and its untangling process leading to the order-preserving solution $\hat{y}$.

Let $\gamma(i) \subseteq S$ be the set of indices that \emph{cross} $i$, that is,
$ i < j \text{ and } y_i > y_j, \text{ or } i > j \text{ and } y_i < y_j$.
Let $\widetilde{\gamma}(i)=\{j \in \gamma(i) : |x_i - y_i| \geq |x_j - y_j|\}$, that is, the set of sensors in $\gamma(i)$ that move at most as far as $i$.
If $y_i<x_i$ we define $h(i)$ to be the $y$-rightmost
sensor in $\widetilde{\gamma}(i)$, and we let $\ell(i)$ be the $y$-rightmost sensor in $\widetilde{\gamma}(i)$ to the left or equal of $x_i$. 
See Figure~\ref{fig:structure}.
Symmetrically, if $y_i\geq x_i$ we define $h(i)$ to be the $y$-leftmost sensor in $\widetilde{\gamma}(i)$, and $\ell(i)$ to be the $y$-leftmost sensor in $\widetilde{\gamma}(i)$ to the right or equal of $x_i$. For sake of brevity, when the interval $i$ is clear from context, we refer to $h(i)$ as $h$ and to $\ell(i)$ as $\ell$. Note that $\ell(i)$ is not well-defined in the case when there are no intervals between $x_i$ and $y_i$.

Let us make some observations about the intervals. Figure~\ref{fig:structure} sums up these observations by depicting $i$ together with $\widetilde{\gamma}(i)$ with $\ell$ and $h$ highlighted.

\begin{figure}[t]
	\centering
		\includegraphics[width=0.7\textwidth]{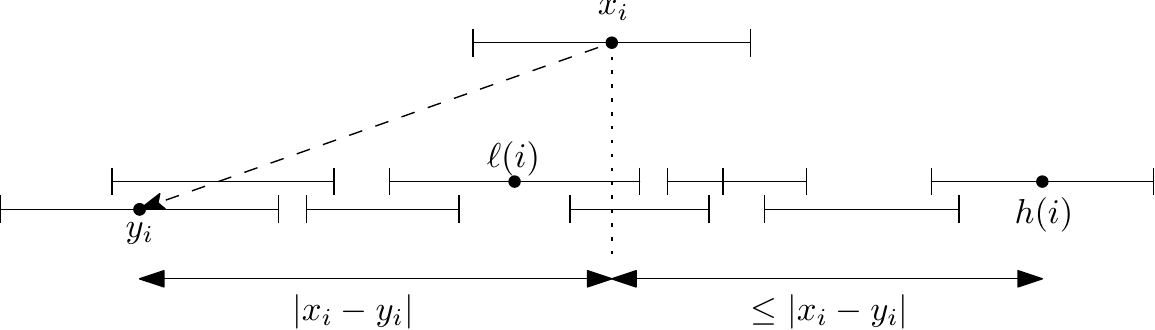}
	\caption{
		\label{fig:structure}
		An interval $i$ and its relation to $\tilde{\gamma}(i)$. In this case $y_i < x_i$, but a symmetric picture holds when $y_i > x_i$.
	}
\end{figure}

\begin{obs}
	\label{obs:close-gamma}
	Every $j \in \widetilde{\gamma}(i)$ must have $y_j \in [x_i - |x_i - y_i|, x_i + |x_i - y_i|]$.
\end{obs}

\begin{proof}
  Note that if $x_i = y_i$ then the claim is trivially true since $\widetilde{\gamma}(i) = \emptyset$.

	Without loss of generality assume $x_i > y_i$, since the case $x_i < y_i$ is symmetric. Since $j \in \widetilde\gamma(i)$ we have $|x_j - y_j| \leq |x_i - y_i|$, and it follows that $x_j < x_i$ and $y_j > y_i$. Therefore, $y_j > y_i = x_i - |x_i - y_i|$ and $y_j < x_j + | x_i - y_i| < x_i + |x_i - y_i|$.
\end{proof}

\begin{obs}
	\label{obs:at-most-two}
	Let $y$ be an optimal solution and let $S$ be a minimal active set of~$y$. Every point stabs (intersects) at most two intervals in $S$.
\end{obs}
\begin{proof}
	If three active intervals in $S$ are stabbed by one point, then one of those intervals can be removed without making the solution infeasible, thus contradicting minimality of $S$.
\end{proof}

\begin{obs}
	\label{obs:overlaps-in-gamma}
	In an optimal solution $y$, if $y_i < x_i$ then the intervals $j \in \widetilde{\gamma}(i)$ such that $y_j > x_i$ do not overlap; similarly, if $y_i > x_i$ then the intervals $j \in \widetilde{\gamma}(i)$ such that $y_j < x_i$ do not overlap.
\end{obs}
\begin{proof}
	Without loss of generality assume $x_i > y_i$, since the case $x_i < y_i$ is symmetric. If there were two indices $j, j' \in \widetilde{\gamma}(i)$ that overlap in $y$ and $y_j > y_{j'} > x_i$, then we could reduce $y_{j'}$ by $r_j + r_{j'} - (y_j - y_{j'})$ to get another feasible solution with lower cost, since $x_j,x_{j'}<x_i$.
	See Figure~\ref{fig:observation} for an illustration.
\end{proof}

\begin{figure}[t]
	\centering
		\includegraphics[width=6cm]{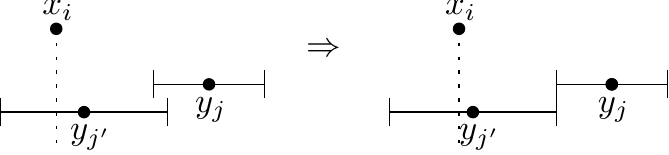}
	\caption{
		\label{fig:observation}
		Illustrating the proof of Observation~\ref{obs:overlaps-in-gamma}, showing that the intervals of $j$ and $j'$ cannot overlap in $y$.
	}
\end{figure}

\begin{obs}
	\label{obs:length-gamma}
	If $\ell$ is well defined for $i$ in an optimal solution $y$ then
	 \[\sum_{j \in \tilde{\gamma}(i)} 2 r_j \leq 3 |x_i - y_i| + r_\ell + r_h.\]
\end{obs}

\begin{proof}
  Note that if $x_i = y_i$ then the claim is trivially true since $\widetilde{\gamma}(i) = \emptyset$.

  Without loss of generality assume $x_i > y_i$, since the case $x_i < y_i$ is symmetric.
	By Observation~\ref{obs:at-most-two} every point in the interval $[y_i, x_i]$ stabs at most two intervals from $\widetilde{\gamma}(i)$.
	By Observation~\ref{obs:overlaps-in-gamma} every point in the interval $[x_i, y_h]$ stabs at most one interval $j \in \widetilde{\gamma}(i)$ such that $y_j>x_i$.
	This accounts for the term $3 |x_i - y_i|$.
	Additionally, we have to add $r_h$ to account for the interval $[y_h,y_h+r_h]$ and $r_\ell$, since $\ell(i)$ might overlap interval $[x_i,x_i+r_\ell]$.
	Let $j$ be the $y$-leftmost sensor in $\widetilde{\gamma}(i)$.
	We do not have to account for the fact that $x_j$ might end left of $y_i$, that is the interval $[y_j-r_j,y_i]$.
	The reason is that $| y_i - y_j + r_j | < r_i$ and counted the interval $[y_i, y_i + r_i]$ already needlessly when considering that $[y_i, x_i]$ stabs at most two intervals from $\widetilde{\gamma}(i)$.
  It follows that
	$\sum_{j \in \tilde{\gamma}(i)} 2 r_j \leq 3 |x_i - y_i| + r_\ell + r_h$.
\end{proof}

\begin{obs}
	\label{obs:card-gamma}
	If $\ell$ is well defined for $i$ in an optimal solution $y$ then
	\[|\widetilde{\gamma}(i)| \leq 3 + \frac{ 3\, |x_i - y_i| - 2\, r_i - r_\ell  - r_h}{2\, r_{\min}}.\]
\end{obs}

\begin{proof}
   From Observation~\ref{obs:length-gamma} we have
   	$\sum_{j \in \tilde{\gamma}(i)} 2 r_j \leq 3 \, |x_i - y_i| + r_\ell + r_h$.
	Notice that each interval in $\tilde{\gamma}(i)$ has length at least $2 r_{\min}$, therefore the number of intervals in $\widetilde{\gamma}(i)$ is no more than $\sum_{j \in \tilde{\gamma}(i)} 2 r_j$ divided by $2 r_{\min}$.
	To get a better bound we count three intervals explicitly: $\ell(i)$, $h(i)$, and $j$, where $j$ is the $y$-leftmost
  sensor if $x_i> y_i$ or the rightmost otherwise.

  Note that if $x_i = y_i$ then the claim is trivially true since $\widetilde{\gamma}(i) = \emptyset$. Without loss of generality assume $x_i > y_i$, since the case $x_i < y_i$ is symmetric.
  Ignoring $j$, we can adjust the bound from Observation~\ref{obs:length-gamma} as follows.
  Since by Observation~\ref{obs:at-most-two} every point stabs at most two intervals, only $j$ might overlap with $i$ in $y$.
  Hence, we only need to consider the interval $[y_i+r_i,x_i]$ where every point stabs at most two intervals from $\widetilde{\gamma}(i)$.
  Hence, ignoring the three explicitly counted intervals, the sum of the lengths of the remaining intervals of $\widetilde{\gamma}(i)$ can be bounded by $2 (|x_i - y_i| - r_i) + |x_i - y_i| + r_\ell + r_h - 2 r_\ell - 2 r_h$. Therefore, we have
	$|\widetilde{\gamma}(i)| \leq 3 + \frac{ 3 |x_i - y_i| - 2 r_i - r_\ell  - r_h}{2 r_{\min}}$.
\end{proof}

Now everything is in place to describe our cost sharing schemes. Our first scheme is simpler to describe and is $(3 \rho + 4)$-balanced. Our second scheme is a refinement and is $\big((2 + \epsilon) \rho + 2/\epsilon \big)$-balanced for any $\epsilon > 0$.

\begin{lemma}\label{lem:scheme-first-bound}
	For an optimal solution $y$ to the barrier problem there is an untangling $\hat{y}$ of $y$ such that there is a $(3 \rho + 4)$-balanced cost sharing scheme.
\end{lemma}
\begin{proof}
	The high level idea of our charging scheme is as follows: When $i$ swaps places with $j \in \tilde{\gamma}(i)$, we charge $i$ enough to pay for the movements of both $i$ and $j$. In particular if $\tilde{\gamma}(i) = \emptyset$ then we do not charge $i$ at all, that is, $\xi(i) = 0$.

	From now on we assume that $\tilde{\gamma}(i) \neq \emptyset$. For the analysis it will be useful to study how $i$ moves in the untangling process. If $y_i < x_i$ then swapping $i$ and $j \in \tilde{\gamma}(i)$ always moves $i$ to the right; similarly, if $y_i > x_i$ then swapping $i$ and $j \in \tilde{\gamma}(i)$ always moves $i$ to the left. On the other hand, when swapping $i$ and $j \in \gamma(i) \setminus \tilde{\gamma}(i)$, the interval $i$ can move either left or right.

		We consider two scenarios. If $\hat{y}_i$ ends up on the same side of $x_i$ as $y_i$ then
	$	  |x_i - \hat{y}_i| \leq \sum_{j\in \gamma(i)\setminus \tilde{\gamma}(i)} 2 r_j  + |x_i - y_i|$,
		so we charge $2r_j$ to each $j \in \gamma(i) \setminus \tilde{\gamma}(i)$ and $|x_i - y_i|$ to $i$. Thus, under this scenario, the total amount charged to $i$ is
		\begin{equation}
			\label{eq:xi-left}
			\xi(i) \leq 2r_i |\widetilde{\gamma}(i)| + |x_i - y_i|
		\end{equation}
		The second scenario is when $\hat{y}_i$ and $y_i$ end up on opposite sides of $x_i$ then
	$	  |x_i - \hat{y}_i| \leq \sum_{j \in \gamma(i)} 2 r_j - |x_i - y_i|$,
		so we charge $\sum_{j \in \tilde{\gamma}(i)} 2 r_j - |x_i - y_i|$ to $i$ and $2r_j$ to each $j \in \gamma(i) \setminus \tilde{\gamma}(i)$. Thus, under this scenario, the total amount charged to $i$ is
		\begin{equation}
			\label{eq:xi-right}
			\xi(i) \leq 2r_i |\widetilde{\gamma}(i)| + \sum_{j \in \widetilde{\gamma}(i)} 2 r_j - |x_i - y_i|.
		\end{equation}

	The rest of the proof is broken up into four cases.

\medskip
\noindent
	Case 1: Intervals $i$ and $h(i)$ overlap in $y$.
	
	In this case $\widetilde{\gamma}(i) = \set{h(i)}$ and $\widetilde{\gamma}(h(i)) = \emptyset$. Furthermore, if there is another interval $i'$ such that $h(i') = h(i)$ then $i'$ and $h(i')$ do not overlap. Indeed, if $y_i$ lies in between $y_{i'}$ and $y_{h(i)}$ then $i'$ and $h(i)$ cannot overlap otherwise there is a point covered by $i$, $i'$ and $h(i)$;
   if $y_{i'}$ lies in between $y_i$ and $y_{h(i)}$ we get a similar contradiction, so it must be that $y_{h(i)}$ lies in between $y_i$ and $y_{i'}$. See Fig.~\ref{fig:case1}. This means that $i$ and $i'$ cross, so either $i \in \widetilde{\gamma}(i')$ or $i' \in \widetilde{\gamma}(i)$, which, together with $h(i)= h(i')$, yields a contradiction.

	Therefore, we can run the untangling process so that all pairs $i$ and $h(i)$ that overlap in $y$ are swapped first. Let $y'$ be the solution after these initial swaps are carried out. Then,
	\begin{align*}
	  |x_i - {y'}_i| + |x_{h(i)} - {y'}_{h(i)}| & \leq |x_i - y_i| + |x_{h(i)} - y_{h(i)}| + 2 |r_i - r_{h(i)}|\\
	                                            & \leq 4 ( |x_i - y_i| + |x_{h(i)} - y_{h(i)}|)
	                                            & \leq 6 |x_i - y_i|.
	\end{align*}
	The first inequality is due to the fact that additional cost comes from swapping $i$ and $h$, where at most one them moves in a direction that increases the cost and they are overlapping.
	Hence the additional cost is bounded by $2 |r_i - r_{h(i)}|$.
	The second inequality  is due to the fact that the movement $|x_i - y_i| + |x_{h(i)} - y_{h(i)}|$ needs to be larger than $|r_i - r_{h(i)}|$ for $i$ and $h$ to swap positions and both be active.
	
	Later on in the untangling process, $i$ and $h$ may be swapped with another interval, call it $j$, causing them to move further and to increase their contribution towards $\mathrm{cost}(\hat{y})$. If this happens, we charge the movement of $i$, or $h$, to $j$. Therefore, setting $\xi(i) = 6 |x_i - y_i|$ is enough to cover the contribution of $i$ and $h$ to the cost of $y$ that is not covered by other intervals. Obviously, the scheme so far is $(3 \rho + 4)$-balanced.

\begin{figure}[t]
	\centering
		\includegraphics[width=6cm]{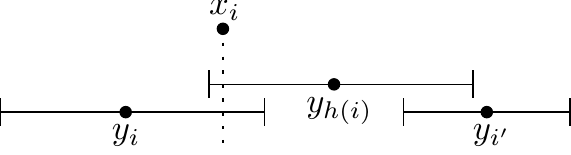}
	\caption{
		\label{fig:case1}
		If $h(i)=h(i')$ then $i$ and $i'$ must lie on opposite sides of $h(i)$ in $y$.
	}
\end{figure}

The proof of Cases 2 and 3 are deferred to the
\longversion{appendix}%
\shortversion{long version~\cite{arxiv}}
where it is shown that when $\ell$ is not well-defined (Case 2) or $\ell$ is well-defined and intervals $\ell$ and $i$ overlap in $y$ (Case 3), then $\frac{\xi(i)}{|x_i-y_i|} \leq 2 \rho + 1$.

\medskip
\noindent
Case 4: $\ell$ is well-defined and  intervals $i$ and $\ell$ do not overlap in $y$.

The assumption implies $|x_i - y_i| \geq r_i + r_\ell$. Since we will use Observation~\ref{obs:length-gamma} to bound $\sum_{j \in \widetilde{\gamma}(i)} 2r_j$, it follows that the sub-case when $i$ is charged the most is when $y_i$ and $\hat{y}_i$ are on opposite sides of $x_i$, so we start with the bound provided by~\eqref{eq:xi-right}:
	\begin{align*}
		\xi(i) & \leq 2r_i |\widetilde{\gamma}(i)| + \sum_{j \in \widetilde{\gamma}(i)} 2 r_j - |x_i - y_i| \\
		       & \leq r_i \left( 6 + \frac{3 |x_i - y_i| - 2r_i - r_\ell - r_h}{r_{\min}} \right) + 2|x_i - y_i| + r_\ell + r_h \\
		       &  = \left( 3 \frac{r_i}{r_{\min}} + 2 + r_i \frac{6 - 2r_i/r_{\min} - r_\ell/r_{\min} - r_h/ r_{\min} + r_h / r_i + r_\ell / r_i}{|x_i - y_i|} \right) |x_i - y_i|  \\
		       &  \leq \left( 3 \frac{r_i}{r_{\min}} + 2 + \frac{4 - 2r_i/r_{\min} + 2 r_{\min} / r_i}{1 + r_{\min}/r_i} \right) |x_i - y_i|
		       \qquad  \leq \left( 3 \frac{r_i}{r_{\min}} + 4\right) |x_i - y_i|\\
		       &  \leq \left( 3 \rho + 4\right) |x_i - y_i|
	\end{align*}
	where the second inequality follows from Observations~\ref{obs:length-gamma} and~\ref{obs:card-gamma}, the third inequality follows from $|x_i - y_i| \geq r_i + r_\ell$, the forth inequality follows from the fact that the right hand side of the previous line decreases with $r_\ell$ and $r_h$, and so it is maximized when $r_\ell = r_h = r_{\min}$, and the fifth inequality follows from the fact that third term inside the parenthesis is a decreasing function for $r_i \geq r_{\min}$.
This completes the proof of Lemma~\ref{lem:scheme-first-bound}.
\end{proof}

\begin{lemma} \label{lem:scheme-bound}
	For an optimal solution $y$ to the barrier problem there is an untangling $\hat{y}$ of $y$ such that there is a $\big((2 +\epsilon) \rho + 2/\epsilon\big)$-balanced charging scheme.
\end{lemma}

\begin{proof}[Proof sketch]
	The key insight to get this charging scheme is to realize that the intervals $j \in \widetilde{\gamma}(i)$ such that $y_i$ and $y_j$ end up on opposite sides of $x_i$ must have $|x_j - y_j| > 0$, so we can use some of this cost to pay for the distance it moves when swapping places with $i$. If $|x_i - y_j| \geq \epsilon |x_i - y_i|$ then swapping $i$ and $j$ causes $j$ to move $2r_i$, we charge that to $j$ instead of $i$ like before.
	In this modified charging scheme $i$ gets charged $(1- \epsilon) \frac{r_i}{r_{\min}} |x_i - y_i|$ less because it does not pay for the movement of $j \in \widetilde{\gamma}(i)$ with $y_j > x_i (1+\epsilon)$. On the other hand, it has to pay for its own movement when swapped with some $j'$ such that $i \in \widetilde{\gamma}(j')$ and $|x_i - y_i| \geq \epsilon |x_{j'} - y_{j'}|$. However, it can be shown that the total extra charge that an interval $i$ is given is at most $\frac{2}{\epsilon} |x_i - y_i|$. Therefore, the scheme is $\big((2 +\epsilon) \rho + 2/\epsilon\big)$-balanced.
\end{proof}

Selecting $\epsilon$ appropriately gives a minimum approximation of $2(\rho + \sqrt{2 \rho})$.
We conclude this sub-section by showing that our analysis is almost tight.

\begin{lemma}
	\label{lem:rho-gap}
	There is a family of instances where the ratio of the cost of the best order-preserving solution to the cost of the unrestricted optimal solution tends to $\rho$.
\end{lemma}

\begin{figure}
	\centering
	\includegraphics[width=10cm]{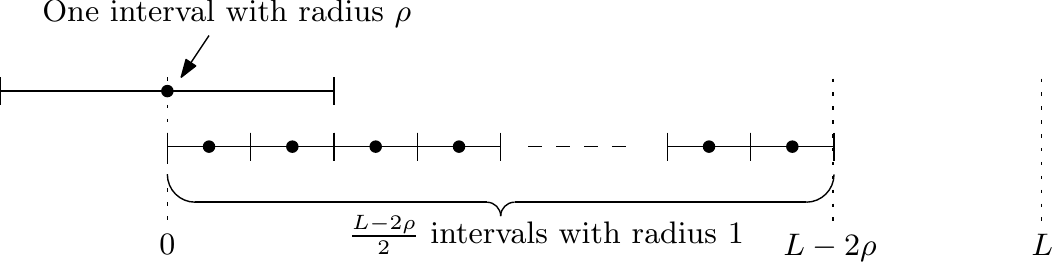}

	\caption{\label{fig:rho_lowerbound} A family of instances showing that order preserving solution cannot guarantee better than $\rho$ approximation. }
\end{figure} %

\begin{proof}
	Consider the instance in Figure~\ref{fig:rho_lowerbound}. There are $\frac{L-2\rho}{2}$ unit-radius intervals covering $[0, L-2\rho]$ and one $\rho$-radius interval covering $[-\rho, \rho]$. The optimal solution moves the long interval $L -\rho$ distance to the right to cover $[L - 2\rho, L]$, at a cost of $L - \rho$. On the other hand, the order-preserving solution involves moving each small interval $2\rho$ units to the right, at a cost of $2\rho \frac{L-2\rho}{2}$. For large enough $L$ the ratio of the cost of these solutions tends to $\rho$.
\end{proof}

As a closing note, we mention that our analysis of the current untangling procedure is nearly tight. Indeed, consider the instance in Figure~\ref{fig:rho_lowerbound2}. The optimal solution moves the long interval $L-\rho$ distance to the right. If there is a small gap between two consecutive small intervals, every interval will be active; therefore, in the untangled solution every small interval is moved a distance of $2\rho$ to the right. This means that the ratio of the cost of the untangled solution to $\opt$ tends to $2 \rho$ as $L$ grows.

\begin{figure}
	\centering
	\includegraphics[width=10cm]{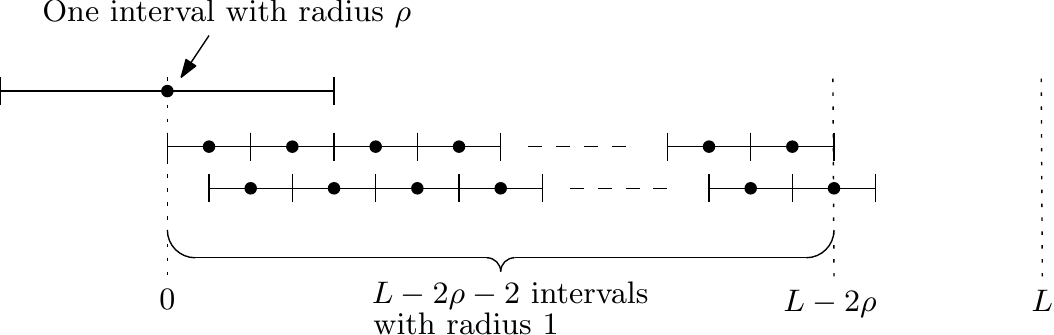}

	\caption{\label{fig:rho_lowerbound2} A family of instances showing that our untangling process can yield solutions that are $2\rho$ away from the optimum. }
\end{figure}

\subsection{Computing Good Order-Preserving Solutions} \label{ssec:algorithm}

First we describe a pseudo-polynomial time algorithm for finding an optimal order-preserving solution. Then we show how to get a $(1+\epsilon)$-approximate order-preserving solution in strongly-polynomial time.

\begin{lemma}
	Assuming the coordinates defining the instance are integral, there is an $O(\opt^2 n )$ time algorithm for computing an optimal order-preserving solution, where $\opt$ is the value of said solution.
\end{lemma}

\begin{proof}
	Consider the following dynamic programming formulation where we let $T[i,b]$ be the largest value such that there is an order-preserving solution using the intervals $1, \ldots, i$ to cover $[0, T[i,b]]$ having cost at most $b$.
	For $i=0$ there is no active set and so $T[0, b] = 0$ for all $b$.
	For $i > 0$, if $i$ is not part of the active set of the solution that defines $T[i,b]$ then $T[i,b] = T[i-1, b]$.
	For $i > 0$, if $i$ is part of the active set in the optimal solution then we can condition on how much $i$ moves, say $k$ units either to the left or to the right. The most coverage that we can possibly get is to move $i$ to $y_i = T[i-1, b-k] + r_i$, which would allow a cover up to $T[i-1, b-k] + 2r_i$; however, this is only possible if $|T[i-1, b-k] + r_i - x_i| \leq k$.  On the other hand, if $|T[i-1, b-k] + r_i - x_i| > k$ then it must be that $x_i < T[i-1, b-k] + r_i$ (otherwise $k$ needs to be larger) and the best coverage we can get is then $x_i + k$, which should be larger than $T[i-1, b-k]$.
	At this point it is straightforward to write a recurrence for $T[i,b]$ that can be computed in $O(b)$ time given the values for $T[i-1, \ast]$. There are $n \times \opt$ dynamic programming states and each takes $O(\opt)$  time to compute.
\end{proof}

\begin{lemma}
	There is an $O(n^3/\epsilon^2)$ time algorithm for computing a $(1+\epsilon)$-approximate order-preserving solution.
\end{lemma}

\begin{proof}
	For $q = \frac{\epsilon \cdot \opt}{n}$ we define the following objective function:
	$
	  \mathrm{cost}'(y) = \sum_{i} \left\lceil{\frac{|y_i - x_i|}{q}}\right\rceil.
	$
	This new cost function is closely related to the original objective, namely:
	$
	  \mathrm{cost}(y) \leq q \cdot \mathrm{cost}'(y) \leq \mathrm{cost}(y) + q n.
	$
	Using the same dynamic formulation as the one used in the pseudo-polynomial time algorithm, we can optimize $\mathrm{cost}'$ in $O(n^3 / \epsilon^2)$ time. Furthermore, the value of this solution under the original objective is at most $(1+\epsilon) \opt$, so the claim follows.
\end{proof}

\section{Inapproximability Results}

The known \NP-hardness proof for the barrier coverage problem~\cite{conf/adhoc/CzyzowiczKKLNOSUY10} is a reduction from 3-{\sc Partition}. The reduction takes an instance of 3-{\sc Partition} and creates an instance of the barrier coverage problem with integral values, $n+1$ different radii values, and $\rho = c n^d$ for some constants $c$ and $d$. Computing a 2-approximate solution in this instance is enough to decide the 3-{\sc Partition} instance. Therefore, there is no 2-approximation unless $\Ptime=\NP$. In fact, the same reduction can be used to obtain inapproximability results in terms of $\rho$.

\begin{theorem} \label{thm:inapproximabilty}
	There is no polynomial time $\rho^{1 - \epsilon}$-approximation algorithm for any constant $\epsilon > 0$ unless $\Ptime=\NP$.
\end{theorem}

\begin{proof}
	As noted %
  in~\cite{conf/adhoc/CzyzowiczKKLNOSUY10}, a similar reduction can be used to construct an instance with $\rho = \alpha c n^d$ for $\alpha > 1$ such that an $\alpha$-approximation is enough to decide the 3-{\sc Partition} instance. If we set $\alpha = (c n^d)^{\frac{1-\epsilon}{\epsilon}}$ then we get that $\alpha = \rho^{1-\epsilon}$ and the claim follows.
\end{proof}

\section{Parameterized Complexity} \label{sec:hardness}

We show that the barrier coverage problem is hard, even if we only allow a small number of sensors to move.
Formally, we show that the following problem is \W{1}-hard when parameterized by $k$.

\begin{center}
\begin{tabular}{|r p{0.8\columnwidth}|}
\hline
\multicolumn{2}{|l|}{\textsc{$k$-move-Barrier-Coverage}} \\
\textit{Instance:} & Sensors $(x_1,r_1),\dots,(x_n,r_n)$, $L\in\mathbb{R}$, $B\in\mathbb{R}$, and $k\in\mathbb{N}$. \\
\textit{Problem:} & Does there exist a barrier coverage $y$ of interval $[0,L]$ such that $\mathrm{cost}(y)\leq B$ and $\abs{\{i \mid x_i \neq y_i\}}\leq k$?\\
\hline
\end{tabular}
\end{center}
To show \W{1}-hardness, we will reduce from \textsc{Exact-Cover}.
\begin{center}
\begin{tabular}{|r p{0.8\columnwidth}|}
\hline
\multicolumn{2}{|l|}{\textsc{Exact-Cover}} \\
\textit{Instance:} & Universe $U=\{u_1,\dots,u_m\}$, set of subsets $S=\{S_1,\dots,S_n\}\subseteq 2^U$, and $k\in \mathbb{N}$. \\
\textit{Problem:} & Does there exist $T=\{T_1,\dots,T_l\}\subseteq S$ such that $l\leq k$, $\bigcup_{i=1}^l T_i = U$, and $T_i\cap T_j =\emptyset$ for $1\leq i < j \leq l$?\\
\hline
\end{tabular}
\end{center}

A special case of \textsc{Exact-Cover} is the problem \textsc{Perfect-Code}, which was shown to be \W{1}-hard when parameterized by $k$~\cite{DowneyF95} (\W{1}-membership was proved later~\cite{Cesati02}).
Hence, \textsc{Exact-Cover} is \W{1}-hard when parameterized by $k$.
Actually, \W{1}-hardness for \textsc{Perfect-Code} was shown for the case where one asks for a solution of size exactly $k$ and not, as in our problem definition, a solution of size at most~$k$.
However, the proof can easily be adapted to our problem variant.

\begin{theorem}
\textsc{$k$-move-Barrier-Coverage} is \W{1}-hard when parameterized by~$k$.
\end{theorem}

\begin{proof}
We reduce from \textsc{Exact-Cover}.
Let $U=\{u_1,\dots,u_m\}$, $S=\{S_1,\dots,S_n\}\subseteq 2^U$, and $k$ be an instance of \textsc{Exact-Cover}.
We construct an instance $(x_1,r_1)$, $\dots$, $(x_n,r_n)$, $L$ and $B$ for \textsc{$k$-move-Barrier-Coverage} as follows.
For $1\leq i \leq n$ and $1\leq j \leq m$ we define
\[
e_{i,j} = \begin{cases}
  (n+1)^{j-1} & \text{ if } u_j \in S_i, \\
  0 & \text{ otherwise.}
\end{cases}
\qquad\qquad
d_{i,j} = \begin{cases}
  (n+1)^{j+m} & \text{ if } u_j \in S_i, \\
  0 & \text{ otherwise.}
\end{cases}
\]
Our instance consists of intervals having radius $r_i = \frac{1}{2} \sum_{j=1}^m e_{i,j}$ and initial position $x_i= -r_i - \sum_{j=1}^m d_{i,j}$ for $1\leq i \leq n$.
Furthermore, we set $L=\sum_{j=1}^m (n+1)^{j-1}$ and $B=\sum_{j=1}^m (n+1)^{j+m} + k \sum_{j=1}^m (n+1)^{j-1}$.
This reduction can be constructed in polynomial time.
Figure~\ref{fig:exact-cover} shows part of the reduction for a small example instance.

For the correctness, first assume that the \textsc{Exact-Cover} instance is a yes-instance, i.e., there exists $T=\{T_1,\dots,T_l\}\subseteq S$ such that $l\leq k$, $\bigcup_{i=1}^l T_i = U$, and $T_i\cap T_j =\emptyset$ for $1\leq i < j \leq l$.
Let $I\subseteq \{1,\dots,n\}$ be the indices of the intervals corresponding to sets $\{T_1,\dots,T_l\}$.
By construction, $\abs{I} \leq k$.
We have to show that $[0,L]$ can be covered by moving only the intervals identified by $I$ and that this solution has cost at most $B$. %
Since $\bigcup_{i=1}^l T_i = U$, for every $u_j \in U$ there exists exactly one $i\in I$ such that $u_j \in S_i$.
Hence, $\sum_{i\in I} r_i = \frac{1}{2} \sum_{i\in I} \sum_{j=1}^m e_{i,j} = \frac{1}{2} \sum_{j=1}^m (n+1)^{j-1}$.
Therefore, the total length of the selected intervals is exactly $L$ and we can cover $[0,L]$.

Next, we consider the cost of this solution.
Moving all the intervals identified by $I$ to the beginning of the barrier, that is, to position $-r_i$ for interval $i\in I$  results in cost $\sum_{j=1}^m (n+1)^{j+m}$.
Again, the argument is that for every $u_j \in U$ there exists exactly one $i\in I$ such that $u_j \in S_i$.
Hence, $\sum_{i\in I} \abs{-r_i-x_i} = \sum_{i\in I} \sum_{j=1}^m d_{i,j} = \sum_{j=1}^m (n+1)^{j+m}$.
Additionally, the movement of these $k$ intervals to the exact position on $L$ can be bounded by $k L$ resulting in a total cost of at most
$\sum_{j=1}^m (n+1)^{j+m} + k \sum_{j=1}^m (n+1)^{j-1}=B$.

For the reverse direction, assume that there exists a barrier coverage $y$ of interval $[0,L]$ such that $\mathrm{cost}(y)\leq B$ and $\abs{\{i \mid x_i \neq y_i\}}\leq k$.
Let $I\subseteq \{1,\dots,n\}$ be the indices of the moved intervals.
We have to show that $T=\{S_i \mid i \in I\}$ is a solution for the \textsc{Exact-Cover} instance, that is, every element $u\in U$ is contained exactly once in the sets of $T$.
Assume towards a contradiction, that this is not true.
Let $u_c \in U$ be the element with the highest index such that $u_c$ is either not contained in $T$ or it occurs more than once.
Since elements $u_{c+1},\dots,u_m$ occur exactly once, they contribute the length $\sum_{j=c+1}^{m}(n+1)^{j-1}$ towards covering $[0,L]$.
Therefore, $\sum_{j=1}^{c}(n+1)^{j-1}$ remains to be covered. We have two cases:
\begin{itemize}
\item {\bf $u_c$ is not contained in $T$.}
Then the maximum length we can cover is if every element $u_1,\dots,u_{c-1}$ is contained in every moved interval.
Since $n \cdot \sum_{j=1}^{c-1}(n+1)^{j-1} = (n+1)^{c-1} - 1 < \sum_{j=1}^{c}(n+1)^{j-1}$, this is not enough and contradicts our assumption that $y$ is a barrier coverage.
Hence, $u_c$ is contained in $T$.
\item {\bf $u_c$ occurs in multiple moved intervals.}
Since elements $u_{c+1},\dots,u_m$ occur exactly once, they  contribute $\sum_{j=c+1}^{m}(n+1)^{j+m}$ to the total cost just for moving the corresponding intervals to the beginning of the barrier.
Since $u_c$ occurs at least twice, it will contribute $2\cdot (n+1)^{c+m}$ to the total cost just for moving the corresponding intervals to the beginning of the barrier.
But $2(n+1)^{c+m} + \sum_{j=c+1}^{m}(n+1)^{j+m} = (n+1)^{c+m} + \sum_{j=c}^{m}(n+1)^{j+m}$, which is larger than our budget $B$, because
$B \leq \sum_{j=1}^m (n+1)^{j+m} + n \sum_{j=1}^m (n+1)^{j-1} < \sum_{j=0}^m (n+1)^{j+m} = \sum_{j=0}^{c-1} (n+1)^{j+m} + \sum_{j=c}^{m} (n+1)^{j+m}$ and $\sum_{j=0}^{c-1} (n+1)^{j+m} < (n+1)^{c+m}$.
Hence, $u_c$ is contained exactly once in the sets of $T$, which contradicts our assumption.
\end{itemize}
Therefore, $T$ is indeed a solution for the \textsc{Exact-Cover} instance.
\end{proof}

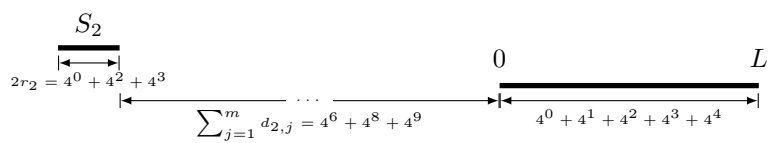
\begin{figure}
\centering
\begin{tikzpicture}

  \draw[line width = 2pt] (-5, 0.5) -- (-5.81, 0.5) node[midway, above]{$S_2$};
  \draw[|<->|, > = latex] (-5, 0.3) -- (-5.81, 0.3) node[midway, below]{\tiny{$2r_2=4^0+4^2+4^3$}};

  \draw[line width = 2pt] (0,0) -- (3.41,0);
  \node[label=$0$] at (0,0) {};
  \node[label=$L$] at (3.41,0) {};
  \draw[|<->|, > = latex] (0,-0.2) -- (3.41,-0.2) node[midway, below]{\tiny{$4^0+4^1+4^2+4^3+4^4$}};

  \draw[|<->|, > = latex] (-5,-0.2) -- (0,-0.2) node[midway,fill=white] {\tiny{$\dots$}} node[midway, below]{\tiny{$\sum_{j=1}^{m} d_{2,j} = 4^6+4^8+4^9$}};
\end{tikzpicture}
    \caption{\label{fig:exact-cover} Part of the reduction from \textsc{Exact-Cover} to \textsc{$k$-move-Barrier-Coverage} for an instance $U=\{u_1,\dots,u_5\}$, $S=\{S_1,S_2,S_3\}$, with $S_2=\{u_1,u_3,u_4\}$.}
\end{figure}

Complementary to this \W{1}-hardness result, we will show next, that the problem is fixed-parameter tractable when parameterized by the budget $B$.
To this end we have to change the problem to restrict the input to integers instead of real numbers.

\begin{center}
\begin{tabular}{|r p{0.8\columnwidth}|}
\hline
\multicolumn{2}{|l|}{\textsc{Barrier-Coverage}} \\
\textit{Instance:} & Sensors $(x_1,r_1),\dots,(x_n,r_n)$ with $x_i,r_i\in \mathbb{N}$ for each $i\in \{1,\dots,n\}$, $L\in\mathbb{N}$, and $B\in\mathbb{N}$. \\
\textit{Problem:} & Does there exist a barrier coverage $y$ of interval $[0,L]$ such that $\mathrm{cost}(y)\leq B$?\\
\hline
\end{tabular}
\end{center}

\begin{theorem}
The \textsc{Barrier-Coverage} problem can be solved in $2^{2B^2(B+1)} \cdot n^{O(1)}$ time. 
\end{theorem}

\begin{proof}
Our algorithm is a branching algorithm, which, for any candidate sensor branches on which integer point in the gaps (empty intervals) to move this sensor to (or leave it at its original position).
The crucial observations will be that we can give a bound on the number of candidate sensors we need to consider to move into the gaps as well as on the positions where they end up in the final configuration, both in terms of the budget $B$.
The sum of the gaps on the barrier is at most $B$, otherwise we have a trivial no-instance.
Given a gap $G$, we only need to consider intervals that are distance $\leq B$ left and right of $G$, since intervals further away cost too much to move them into $G$.
Assume the interval of $G$ is $[y_l,y_r]$.
We consider the range left of $G$, that is $[y_l-B,y_l]$ (the right side is symmetrical).
At each point $p_i$ in $[y_l-B,y_l]$, we consider all the intervals whose right end equals $p_i$, that is intervals $(x_j,r_j)$ with $x_j+r_j = p_i$.
Let $S_i$ denote the set of these intervals.
We would like to branch on which intervals (if any) from $S_i$ move into the gap $G$, but $|S_i|$ is not necessarily bounded by a function of $B$.
Hence, we sort the intervals in $S_i$ by length and consider only the $B+1$ longest ones.
This is sound, since our budget allows us to move at most $B$ intervals and additionally, an interval from $S_i$ might need to remain stationary in order to cover $p_i$.
Assume there exists an optimal solution in which interval $(x_j,r_j)\in S_i$ is moved to position $y_j \ne x_j$ and $(x_j,r_j)$ is not among the top $B+1$ longest ones.
Then at most $B-1$ of the longest intervals in $S_i$ where moved.
This leaves at least two remaining intervals among the $B+1$ many.
Assume $(x_k,r_k)$ is the shorter one of those two.
Moving $(x_k,r_k)$ the same distance to the right as $(x_j,r_j)$ was moved, covers everything $(x_j,r_j)$ was covering and has the same cost.
Additionally, $[x_k-r_k,x_k+r_k]$ is still covered by the longer interval which we did not move.
Hence, to conclude, we need to consider at most $B+1$ intervals for each of the $B$ points left and right of a gap.

The only thing remaining, is to show that it suffices to consider integer points for the solution.
\longversion{By Lemma~\ref{lem:integer-solution} in the appendix, this is indeed the case.}%
\shortversion{The proof of this is deferred to the long version~\cite{arxiv}}.
Therefore, for our branching algorithm, the total number of intervals to consider is bounded by $B$ and their possible new positions is bounded by the budget $B$ as well, which leads to fixed-parameter tractability in $B$ because $B$ decreases by at least one in each recursive call.
\end{proof}

\section{Conclusion}
We showed a $((2+\varepsilon)\rho+2/\varepsilon)$-approximation for the barrier coverage problem for the case when the sensors initially are on a line containing the barrier.
This works well when the ratio between the largest radius and the smallest radius is small, but in theory the difference could be arbitrarily large.
However, we also proved that no polynomial time $\rho^{1-\varepsilon}$-approxi\-mation algorithm exists unless $\Ptime=\NP$.
There are still several open problems for this special case that would be interesting to pursue.

\begin{enumerate}
	\item Improve the approximation ratio analysis of an order-preserving solution. Ideally, down to $\rho + O(1)$.
	\item Determine if the problem is fixed-parameter tractable for parameter $k$ when the interval radii are $1, 2, \ldots, k$.
	\item Approximate the weighted version where each interval has a weight and we want to minimize $\sum_i w_i |x_i - y_i|$.
\end{enumerate}

\longversion{

\appendix
\newpage

\section{Missing proofs}

\begin{proof}[Proof complement of Lemma \ref{lem:scheme-first-bound}: Cases 2 and 3]
	We will now show that $\frac{\xi(i)}{|x_i-y_i|} \leq 2 \rho + 1$ when either $\ell$ is not well-defined (Case 2) or $\ell$ is well-defined and intervals $\ell$ and $i$ overlap in $y$ (Case 3), complementing the proof that the charging scheme is $(3 \rho + 4)$-balanced.
	
\paragraph*{Case 2: $\ell$ is not well-defined.}
Assume $i$ a $h(i)$ do not overlap, otherwise we are in Case~1. This means $|x_i - y_i| \geq r_h$. If $\widetilde{\gamma}(i) = \set{h}$, it is easy to see that $ \frac{\xi(i)}{|x_i - y_i|} \leq \frac{2r_i + 2 r_h}{r_h} \leq \rho + 1$, so let us assume the stronger property that $\abs{\widetilde{\gamma}(i)} \geq 2$.

Since $\ell$ is not well-defined,  there are no intervals in $\widetilde{\gamma}(i)$ that lie (in $y$) between $x_i$ and $y_i$. It follows that there are at least two intervals in $\widetilde{\gamma}(i)$ that lie on the side of $x_i$ opposite to $y_i$; let $j$ be interval with  $y_j$ closest to $x_i$. By Observation~\ref{obs:overlaps-in-gamma}, $j$ and $h$ and all the other intervals in between them cannot overlap. Using this fact, we can derive two inequalities:
$$
  |x_i - y_i| \geq r_j + r_h + 2 r_{\min} (|\widetilde{\gamma}(i)| - 2) \geq 2 r_{\min} (|\widetilde{\gamma}(i)| - 1)\text{, \quad and}
$$
$$
  \sum_{j \in \widetilde{\gamma}(i)} 2r_j \leq r_j + |x_i - y_i| + r_h.
$$

Therefore, using~\eqref{eq:xi-left}~and~\eqref{eq:xi-right} we get
\begin{align*}
  \frac{\xi(i)}{|x_i-y_i|}
  & \leq \frac{2 |\widetilde{\gamma}(i)| r_i + \max \{ |x_i - y_i|, \sum_{j \in \widetilde{\gamma}(i)} 2r_j - |x_i-y_i| \}}{|x_i - y_i|} \\
  & \leq \frac{2 |\widetilde{\gamma}(i)| r_i}{|x_i - y_i|} + \max \left\{ 1, \frac{ r_j + r_h}{|x_i - y_i|} \right\} \\
  & \leq \frac{2 |\widetilde{\gamma}(i)| r_i}{2r_{\min} (|\widetilde{\gamma}(i)| - 1)} + \max \left\{ 1, \frac{ r_j + r_h}{ r_j + r_h} \right\} \\
  & \leq \frac{|\widetilde{\gamma}(i)| \rho}{ (|\widetilde{\gamma}(i)| - 1)} + \max \left\{ 1, \frac{ r_j + r_h}{ r_j + r_h} \right\} \\
  & \leq 2 \rho + 1,
\end{align*}
where the second to last inequality follows from the fact that the previous expression is maximized when $r_i = r_{\max}$, and the last inequality, from the fact that the previous expression is a decreasing function of $\abs{\widetilde{\gamma}(i)}$, so the maximum value is attained at $\abs{\widetilde{\gamma}(i)}=2$.
Therefore, the charging scheme so far is $(3 \rho + 4)$-balanced.

\begin{figure}[t]
\centering
  \includegraphics[width=6cm]{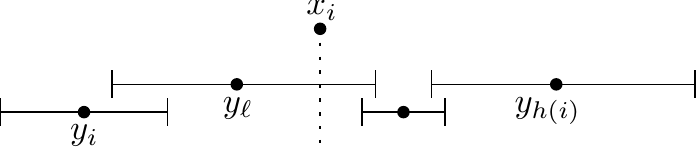}
\caption{
  \label{fig:case3}
      If $\ell$ is well-defined and overlaps $i$ then there can be no other sensor from $\widetilde{\gamma}(i)$ can lie on the same side of $x_i$ as $y_i$.}
\end{figure}

\paragraph*{Case 3: $\ell$ is well-defined and intervals $\ell$ and $i$ overlap in $y$.} Since $\ell$ is well-defined and overlaps $i$, there can be no other sensor from $\widetilde{\gamma}(i)$ can lie (in $y$) on the same side of $x_i$ as $y_i$, see Fig.~\ref{fig:case3}.

Notice that because $\ell$ exists, the interval $i$ cannot overlap any interval that lies (in $y$) to right of $\ell$. Using this fact, we can derive the following inequality:
$$
  2|x_i - y_i| \geq r_i + \sum_{\mathclap{j \in \widetilde{\gamma}(i) \setminus \set{\ell, h}}} 2r_j + r_h
$$

Therefore, if we are under the regime of~\eqref{eq:xi-left}~then we have
\begin{align*}
  \frac{\xi(i)}{|x_i-y_i|}
  & \leq \frac{2 |\widetilde{\gamma}(i)| r_i + |x_i - y_i| }{|x_i - y_i|} \\
  & \leq \frac{4 |\widetilde{\gamma}(i)| r_i}{r_i + 2r_{\min} (|\widetilde{\gamma}(i)| - 2) + r_h} + 1 \\
  & \leq \frac{4 |\widetilde{\gamma}(i)| \rho}{\rho + 2 |\widetilde{\gamma}(i)| - 1} + 1 \\
  & \leq 2 \rho + 1, \\
\end{align*}
where the second to last inequality follows from the fact that the previous expression is maximized when $r_i = r_{\max}$ and $r_h = r_{\min}$, and the last inequality, from the fact that the previoius expression increases as $\abs{\widetilde{\gamma}(i)}$ increases, so the maximum is attained when  $\abs{\widetilde{\gamma}(i)} \rightarrow \infty$.

Finally, if we are under the regime of~\eqref{eq:xi-right} then we have
\begin{align*}
  \frac{\xi(i)}{|x_i-y_i|}
  & \leq \frac{2 |\widetilde{\gamma}(i)| r_i + \sum_{j \in \widetilde{\gamma}(i)} 2r_j - |x_i - y_i| }{|x_i - y_i|} \\
  & \leq \frac{4 |\widetilde{\gamma}(i)| r_i + 2 \sum_{j \in \widetilde{\gamma}(i)} 2r_j}{r_i + \sum_{j \in \widetilde{\gamma}(i) \setminus \set{\ell, h}} 2r_j + r_h} - 1 \\
  & \leq \frac{(4 |\widetilde{\gamma}(i)|-2) r_i + 4 r_\ell + 2r_h+ 2 \left( r_i + \sum_{j \in \widetilde{\gamma}(i) \setminus \set{\ell, h}} 2r_j +r_h\right)}{r_i + \sum_{j \in \widetilde{\gamma}(i) \setminus \set{\ell, h}} 2r_j + r_h} - 1 \\
  & \leq \frac{(4 |\widetilde{\gamma}(i)|-2) r_i + 4 r_\ell + 2r_h}{r_i + \sum_{j \in \widetilde{\gamma}(i) \setminus \set{\ell, h}} 2r_j + r_h} + 1 \\
  & \leq \frac{(4 |\widetilde{\gamma}(i)|+2) \rho + 2}{\rho + 2 \abs{\widetilde{\gamma}(i)} -3} + 1 \\
  & \leq 2 \rho + 3, \\
\end{align*}
where the second to last inequality follows from the fact that the previous expression is maximized when $r_i = r_\ell = r_{\max}$ and the remaining intervals have radius $r_{\min}$, and the last expression increases as a function of $|\widetilde{\gamma}(i)|$ for $|\widetilde{\gamma}(i)|\geq 4$.

Therefore, the charging scheme so far is $(3 \rho + 4)$-balanced.
\end{proof}

\begin{lemma}\label{lem:integer-solution}
	There is an optimal solution $y$ for the \textsc{Barrier-Coverage} problem where each new sensor position $y_i\in y$ is an integer.
\end{lemma}
\begin{proof} %
	We will show that any optimal solution can be converted into an optimal solution with no sensors at non-integral positions.
	The proof is by induction on the number $f$ of sensors at non-integral positions in an optimal solution $y$.
	
	If $f=0$ we are done.
	Suppose that if there is an optimal solution with at most $f-1$ sensors at non-integral positions, then there is an optimal solution with no sensors at non-integral positions.
	Let $\epsilon_l$ be the smallest distance that any non-integral sensor has to move to the left to become integral.
	Let $\epsilon_r$ be the smallest distance that any non-integral sensor has to move to the right to become integral.
	Among the non-integral sensors, either at least half of them have their movement cost reduced by moving to the left or more than half of them have their movement cost reduced by moving to the right.
	Therefore, consider the following two solutions with at most $f-1$ non-integral sensors: in the first one, all non-integral sensors are moved a distance of $\epsilon_l$ to the left, and in the second one, all non-integral sensors are moved a distance of $\epsilon_r$ to the right.
	At least one of these two solutions has cost at most $\textrm{cost}(y)$. Moreover, it is easy to see that the barrier is covered in both solutions. Therefore, by our induction hypothesis, there is an optimal solution with no sensors at non-integral positions.
\end{proof}

} %

\end{document}